%% file: paragraph.tex
\documentclass{llncs}

\usepackage{amsmath}
\usepackage{cleveref}
\usepackage{enumerate}
\newtheorem{observation}{Observation}
\usepackage{subcaption}
\usepackage{graphicx}
\usepackage{multicol, latexsym, amssymb}
\usepackage{blindtext}
\usepackage{caption}
\usepackage[misc,geometry]{ifsym}
\usepackage{xcolor}

\usepackage[numbers,sort&compress,sectionbib]{natbib} 

\AtEndEnvironment{proof}{\phantom{}\qed} 

\newcommand{\miniparagraph}[1]{\textbf{#1}\; }
\newcommand{\miniheading}[1]{\noindent \textbf{#1}.}

\title{Maximum Flows in Parametric Graph Templates}

\author{Tal Ben-Nun\inst{1}\orcidID{0000-0002-3657-6568} \and Lukas Gianinazzi (\Letter)  \inst{1} \orcidID{0000-0001-5975-4526} \and Torsten Hoefler\inst{1}\orcidID{0000-0002-1333-9797} \and Yishai Oltchik\inst{2}}
\institute{ETH Zurich, Universit{\"a}tstrasse 6, Z{\"u}rich \\
\email{\{talbn,glukas,htor\}@inf.ethz.ch} 
\and NVIDIA \email{yoltchik@nvidia.com}}

\begin{document}

\maketitle

\begin{abstract}
Execution graphs of parallel loop programs exhibit a nested, repeating structure. 
We show how such graphs that are the result of nested repetition can be represented by succinct parametric structures. 
	This \emph{parametric graph template} representation allows us to reason about the execution graph of a parallel program at a cost that only depends on the program size.	
	We develop structurally-parametric polynomial-time algorithm variants of maximum flows. When the graph models a parallel loop program, the maximum flow provides a bound on the data movement during an execution of the program. By reasoning about the structure of the repeating subgraphs, we avoid explicit construction of the instantiation (e.g., the execution graph), potentially saving an exponential amount of memory and computation. 
Hence, our approach enables graph-based dataflow analysis in previously intractable settings.
\end{abstract}


\setcounter{page}{1}
\input{sec-intro}

\input{sec-intro-relwork-contr}

\input{sec-definitions}

\input{sec-applications}

\input{sec-stflow}

\input{sec-sibling}

\input{sec-conclusion}

\vspace{0.5em}
\miniheading{Acknowledgements}
This work received support from the PASC project DaCeMI and from the European Research
Council under the European Union's Horizon 2020 programme (Project PSAP, No. 101002047), as well as funding from EuroHPC-JU under grant DEEP-SEA, No. 955606.

\bibliographystyle{splncs04}
\bibliography{parabib.bib}

\end{document}

%% file: sec-intro.tex
\section{Introduction}

\begin{figure}[b]
\begin{subfigure}[t]{0.35\linewidth}
\centering
	\includegraphics[width=0.75\textwidth]{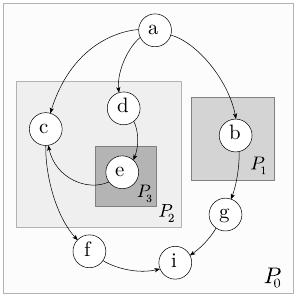}
	\caption{Param. Graph Template $\mathcal{G}$.} \label{fig:intro-example:template}
\end{subfigure}
\quad
\begin{subfigure}[t]{0.27\linewidth}
\centering
\raisebox{2em}{
	\includegraphics[width=0.45\textwidth]{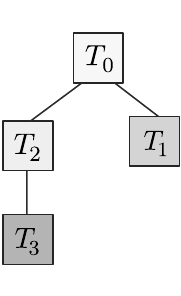}}
	\caption{Template tree of $\mathcal{G}$.} \label{fig:intro-example:tree}
\end{subfigure}
\quad
\begin{subfigure}[t]{0.30\linewidth}
	\includegraphics[width=\textwidth]{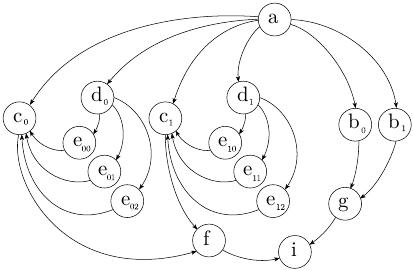}
	\caption{Instantiation of $\mathcal{G}$.} \label{fig:intro-example:instance}
\end{subfigure}
\caption{Illustration of the parametric graph template $\mathcal{G}$ with templates $T_0=\{a, b, c, d, e, f, g, i\}$, $T_1=\{b\},T_2=\{c, d, e\}, T_3=\{e\}$; parameters $P_0=1, P_1=2, P_2=2, P_3=3$; and $h=2$. 
}\vspace{-0em}
\label{fig:intro-example}
\end{figure}

Parallel program analysis approaches to optimize data movement and program transformation commonly rely on graph algorithms~\cite{DBLP:journals/computer/ChuHLE80,DBLP:journals/tc/ShenT85,DBLP:conf/spaa/KwasniewskiBGCS21,llvm}. These problems concern an execution graph, where vertices model computation and edges model data movement. Maximum flows provide an algorithmic measure of the overall data movement. Such execution graphs contain highly repetitive substructures. Other application areas also face repeating graph structures, for example, computational biology~\cite{Ginsburg2016} and network topology~\cite{dragonfly,slimfly,xpander}.

%

The naive approach is to directly work on the graphs and apply classic algorithms. However, this is prohibitively slow. For example, execution graphs can have billions of vertices or have a parametric size. Another approach is to design domain-specific representations and solutions~\cite{DBLP:journals/corr/abs-1802-04730,DBLP:journals/ijpp/Feautrier92}. Having a more general-purpose framework would allow sharing progress across domains. 

We observe that many application-relevant graphs follow a model of \emph{nested repetition}, where a small \emph{template graph} is repeated a \emph{parametric} number of times~\cite{sdfg}. In this work, we propose a representation of such hierarchically repeating graphs, which we call \emph{parametric graph templates}, and provide algorithms for extensions to the classical graph problem of maximum $s$-$t$ flow.

The main challenge lies in avoiding the naive solution of materializing the full graph (which we call \emph{instantiation}) and using a classic algorithm, which would negate any time savings.
%
Instead, we carefully study the structural relationship between the template and the potentially exponentially larger instantiated graph. We discover and exploit symmetries in the instantiation process. This allows us to answer graph problems with a runtime that only depends on the size of the succinct representation, enabling asymptotic time and space savings compared to a naive approach that explicitly performs the nested repetition.





\subsection{Parametric Graph Templates}

Next, we introduce our model and give some examples. Our goal is to represent graphs with a hierarchically repeating structure, where the number of repetitions depends on some parameters. This will allows us to represent parallel loop programs and their executions.
A parametric graph template with $k$ parameters $\mathcal{G}=(G, \mathcal{T}, \mathcal{P})$ contains a (potentially weighted) and directed \emph{template graph} $G=(V, E)$ with $n$ vertices $V$, $m$ edges $E$ and edge weights $w: E\mapsto \mathbb{R}$, a list of \emph{templates} $\mathcal{T}=T_0, T_2, \dotsc, T_{k-1}$, each with $\emptyset \ne T_i \subseteq V$, and a list of positive integer \emph{parameters} $\mathcal{P}=P_0, \dotsc, P_{k-1}$ (see \Cref{fig:intro-example:template}).
The templates follow a nested structure, meaning that for every pair of templates they are either disjoint or one of them is strictly contained in the other one (for all $i\neq j$, $T_i \cap T_j = \emptyset$ or  $T_i \subset T_j$ or $T_j \subset T_i$). In particular, the templates form a laminar set family~\cite{DBLP:conf/esa/CheriyanJR99}.

We assume that there is a \emph{root template} $T_0=V$. Hence, the subset relation on the templates induces a \emph{template tree} (see \Cref{fig:intro-example:tree}). We denote its height by $h$. If a template $T$ is contained in another template $T'$ (i.e., $T\subset T'$), then $T$ is a descendant of $T'$ (and $T'$ is an ancestor of $T$). A template $T$ is a parent of $T'$ (and $T'$ is a child of $T$) if $T'$ is the direct descendant of $T$. 

To create an \emph{instantiation} of a parametric graph template $\mathcal{G}$, repeatedly rewrite it as follows (see \Cref{fig:intro-example:instance}). As long as there is more than one template, pick a leaf template $T_i$. For each vertex $v$ in $T_i$ create $P_i$ copies $v_1, \dotsc, v_{P_i}$ called instances of $v$, replacing $v$ in $V$. The set of vertices with the same subscript are called an \emph{instance} of $T_i$. For each edge $e=(u, v)$ with both endpoints in $T_i$, create $P_i$ instances $e_1=(u_1, v_1), \dotsc, e_{P_i}=(u_{P_i}, v_{P_i})$, replacing $e$ in $E$. For each edge $e=(u, v)$ with one endpoint $u$ in $T_i$, create $P_i$ instances $e_1=(u_1, v), \dotsc, e_{P_i}=(u_{P_i}, v)$, replacing $e$ in $E$. Proceed symmetrically for each edge $e=(u, v)$ with one endpoint $v$ in $T_i$. Then, remove the template $T_i$ and its parameter $P_i$. 

In \Cref{sec:loop-parallel-programs} we will represent nested parallel loop programs as parametric graph templates. In such a representation, cuts and flows correspond to data movement in a parallel execution.

%% file: sec-intro-relwork-contr.tex
 \subsection{Related Work} 

\miniparagraph{Graph Grammars}\cite{DBLP:conf/stacs/Courcelle88,DBLP:conf/focs/EhrigPS73,DBLP:journals/jacm/Pavlidis72,DBLP:journals/mst/BauderonC87,DBLP:conf/fct/Engelfriet89} describe a (possibly infinite) language of graphs compactly with a set of construction rules. There is a wide variety of such ways of constructing a graph, differing in expressive power.  A classic problem for graph grammars is to decide whether a graph can be constructed from a given grammar (parsing). In contrast to graph grammars, we are not primarily concerned with expressing an infinite set of graphs, but instead with a succinct representation of a graph and algorithmic aspects of solving graph problems efficiently on this succinct representation.

\miniparagraph{Hierarchical Graphs}\cite{DBLP:journals/jcss/DrewesHP02} model graphs where edges expand to other, possibly hierarchical graphs. They are a variant of context-free hyperedge replacement grammars that incorporate a notion of hierarchy. The authors consider graph transformations (i.e., replacing subgraphs within other subgraphs). However, their method does not include parametric replication. This makes it unsuitable for modeling variably-sized execution graphs. 

\miniparagraph{Nested graphs}\cite{DBLP:journals/tois/PoulovassilisL94} allow ``hypernodes'' to represent other nested graphs. The authors focus on the case where a node represents a \emph{fixed} nested graph. This precludes nested graphs from effectively representing graphs of a parametric size. 
%

\miniparagraph{Edge-Weight Parametric Problems}
Several graph problems have been generalized to the edge-weight parametric case, where edge weights are functions of one or several parameters $\mu_i$. This includes maximum $s$-$t$ flow / minimum $s$-$t$ cut~\cite{DBLP:journals/dam/AnejaCN03,DBLP:journals/mp/GranotMQT12,DBLP:journals/siamcomp/GalloGT89}, (global) minimum cut~\cite{DBLP:journals/mp/AissiMMQ15,DBLP:conf/stoc/Karger16} and shortest paths~\cite{DBLP:conf/soda/Erickson10,DBLP:journals/dam/KarpO81}. The solution is then a piecewise characterization of the solution space. Usually, only linear (or otherwise heavily restricted) dependency of the edge weights on the parameters have been solved.
%
%
For \emph{edge-weight parametric minimum $s$-$t$ cuts}, the problem can be solved in polynomial time when each edge $e$ has weight $\min(c(e), \mu)$ for constants $c(e)$ and a single parameter $\mu$~\cite{DBLP:journals/dam/AnejaCN03}. Granot, McCormick, Queyranne, and Tardella explore other tractability conditions~\cite{DBLP:journals/mp/GranotMQT12}.

\subsection{Problem Statement}
We approach parametric graph templates from an algorithmic perspective. The goal is to solve classical graph problems for \emph{fixed parameters}, but in time that is strongly polynomial in the size of the parametric graph template. We focus on the classic problem of maximum $s$-$t$ flow, which has an interpretation in terms of data movement for program-derived graphs and operations research~\cite{10.5555/137406}.
%
For an execution graph, a maximum $s$-$t$ flow corresponds to a upper bound on the data movement between vertices $s$ and $t$ when they are placed on different processors. 
 
 \emph{An $s$-$t$ flow} $f$ assigns every edge $e$ a nonnegative real flow $f(e) \leq w(e)$. The sum $\sum_{e=(u, v)}f(e) - \sum_{e=(v, w)}f(e)$ is the \emph{net flow} of the vertex $v$. A flow has to have net flow $0$ for all vertices except $s$ and $t$. The value of the flow is the net flow of the source. A maximum flow is a flow of maximum value.

The maximum $s$-$t$ flow problem has a natural generalization to parametric graph templates when $s$ and $t$ are vertices in the root template: Instantiate the graph and compute a maximum flow between the only instance of $s$ and the only instance of $t$. 
There are multiple possibilities for how to interpret the case when $s$ and $t$ have multiple instances. One interpretation is as a multiple-source and multiple-target flow problem, where all instances of $s$ are treated as sources and all instances of $t$ as sinks. We call this a \emph{maximum all-$s$-$t$ flow}. 
Another interpretation considers the maximum flow between a fixed instance of $s$ and a fixed instance of $t$. We call this a \emph{maximum single-$s$-$t$ flow}. 

\subsection{Our Results}

We show how to efficiently represent a class of parallel loop programs as parametric graph templates and how properties of data movement in the parallel loop programs relates to cuts and flows in the parametric graph templates.

Then, we demonstrate that maximum $s$-$t$ flow  can be solved asymptotically faster than instantiating the parametric graph template. In particular, it is possible to obtain a runtime that is similar to the runtime on the template graph. 

For maximum all-$s$-$t$ flow, our algorithms match the runtime of a maximum $s$-$t$ algorithm such as Orlin's $O(mn)$ time algorithm~\cite{DBLP:conf/stoc/Orlin13}. We solve this problem using a technique called \emph{Edge Reweighting}. It observes that scaling the edge weights in the graph template solves the problem.
For maximum single-$s$-$t$ flow and minimum cuts, there is an overhead proportional to the height $h$ of the template tree. In addition to Edge Reweighting, we use a technique called \emph{Partial Instantiation}. We observe that a carefully chosen part of the instantiated graph can give sufficient information to extrapolate the result to the rest of the graph. How this part is chosen depends on the problem.

%% file: sec-definitions.tex
\section{Preliminaries} \label{sec:definitions}

We proceed to introduce definitions, notation, and assumptions that we use throughout this work. 

\miniheading{Template a vertex belongs to} If a vertex $v$ is in a template $T_i$ and $v$ is in no other template that is a descendant of $T_i$, then $v$ \emph{belongs to} $T_i$. We denote the unique template that $v$ belongs to by $T(v)$. 

\miniheading{Template an edge belongs to}
If both endpoints of an edge belong to a template $T_i$, then this edge belongs to template $T_i$. We denote the number of vertices and edges that belong to a template $T_i$ by $n_i$ and $m_i$, respectively.

\miniheading{Cross-template edges} An edge $(u, v)$ where $u$ and $v$ belong to different templates is \emph{cross-template}.

\miniheading{No Jumping} We assume there are no edges that `jump' layers in the template hierarchy. Specifically, if $(u,v)$ is a cross-template edge, then $T(u)$ is a parent or child of $T(v)$. This rule ensures that a path in the graph corresponds to a walk in the template tree. It comes without loss of generality for cut and flow problems, as an edge that jumps layers can be split into multiple edges (all of weight $\infty$ except the edge connected to the vertex that belongs to the deeper template in the template tree). For graphs that model programs, this assumption corresponds to disallowing jumps to arbitrary program locations.

\miniheading{Boundary Vertices}
Consider a vertex $u$ and $v$ where $T(v)$ is a parent of $T(u)$. If there is an edge from $u$ to $v$ or from $v$ to $u$ in the template graph, then $v$ is a \emph{boundary vertex} of $T(u)$.

\miniheading{Template graph of a template}
The subgraph of the template graph $G$ induced by a template $T_i$ is called the \emph{template graph of} $T_i$. 

\miniheading{Instance tree} We extend the nomenclature of templates to instances. The template hierarchy can be transferred onto the instances, where an instance $I$ is a descendant of an instance $I'$ if the template $T$ that instantiated $I$ is a descendant of the template $T'$ that instantiated $I'$. Similarly, we extend the notions of ancestor, parent, and child to the instances, creating an \emph{instance tree}. Two instances that have the same parent instance are \emph{siblings}.
If a vertex $v$ is contained in an instance $I$, but it is not contained in any other descendant of $I$, the vertex $v$ belongs to the instance $I$. 
If $b_i$ is an instance of a boundary vertex $b$ of a template $T$, then $b_i$ is a boundary vertex of the instance that $b_i$ belongs to. The instance of the root template is the \emph{root instance}. For a vertex $v$ in the instantiation, we write $T(v)$ for the template of the instance that $v$ belongs to.
 
\miniheading{Isomorphism} Two parametric graph templates $\mathcal{G}_1$ and $\mathcal{G}_2$ are \emph{isomorphic} if they instantiate isomorphic graphs. Two isomorphic $\mathcal{G}_1$ and $\mathcal{G}_2$ can have different parameters, templates, and their template graphs need not be isomorphic.

\miniheading{Cycles}
Acyclic graphs are easier to handle for many algorithmic problems. In parametric graph templates, we consider two different notions of what constitutes a cycle.
The simplest notion of cycles comes from considering cycles in the template graph. If it does not contain any cycles, then the instantiation does neither (and vice versa). 
%
A path $p_1, \dotsc p_k$ in the template graph that contains three vertices $p_i, p_j, p_k$ with $i<j<k$ and $T(i)=T(k)$ but $T(i) \neq T(j)$ is a template-cycle. We say a parametric graph template is \emph{template-acyclic} if it does not contain a template-cycle. This notion is incomparable to the notion of acyclic parametric graph templates. 
There are acyclic parametric graph templates that are not template-acyclic (consider a path whose nodes alternate between belonging to some template and its child). Note that a template-acyclic graph can have cycles (consider a cycle whose vertices belong to the same template).

%% file: sec-applications.tex



\section{Templates of Parallel Loop Programs}\label{sec:loop-parallel-programs}

We show that a broad class of parallel programs can be modeled as parametric graph templates, such that the parametric graph template corresponds to the source code of the program and an instantiation of the parametric graph template corresponds to an execution of the program. This allows us to analyze properties of the execution of a program by considering a parametric graph template of a size comparable to the source of the program.  

The parametric graph templates we consider can model nested loop programs, for example Projective Nested Loops~\cite{DBLP:conf/spaa/DinhD20} and Simple Overlap Access Programs~\cite{DBLP:conf/spaa/KwasniewskiBGCS21}. The program receives a set of multi-dimensional input arrays $A_1, \dotsc, A_k$. The goal is to output a multi-dimensional array $B$. The program can use several multi-dimensional temporary arrays $C_1, \dotsc C_{k'}$. For any array $D$, its size in the $i$-th dimension is $\text{size}_i (D)$. 

Roughly speaking, we allow any composition of elementary operations and parallel nested loops where the loop bounds only depend on the sizes of the input arrays. We allow parallel reduction to aggregate the results of a loop. We do not allow data-dependent control flow, but we allow the locations of memory accesses to be data-dependent. Examples of algorithms that can be represented this way include matrix multiplication, convolution, and cross-correlation. See \Cref{fig:examples:mm} for a parametric graph template of matrix multiplication and \Cref{fig:examples:conv} for a parametric graph template of cross-correlation. 

We call the resulting parametric graph templates \emph{parallel loop graph templates}.
Next, we describe their syntax. Then, we describe a semantic for these parametric graph templates. Finally, we relate the data movement of the parallel loop programs with their templates' instantiations.

\begin{figure}[b!]
\begin{subfigure}[t]{0.535\linewidth}
\centering
\includegraphics[width=0.8\linewidth]{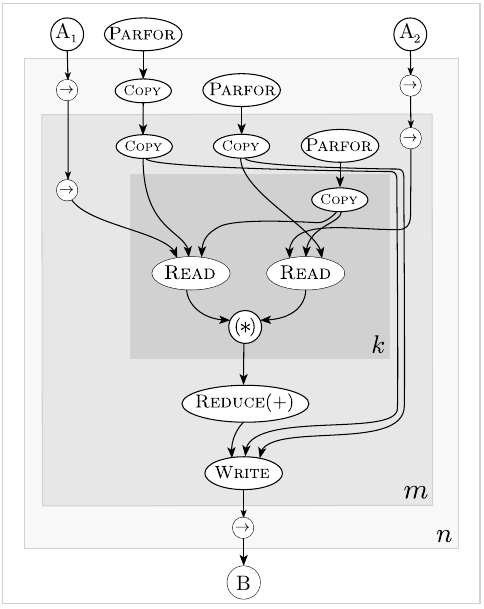}
\caption{Multiplication of an $n\times k$ matrix $A_1$ and a $k\times m$ matrix $A_2$. }
\label{fig:examples:mm}
\end{subfigure}
\hfill
\begin{subfigure}[t]{0.43\linewidth}
\centering
\includegraphics[width=0.8\linewidth]{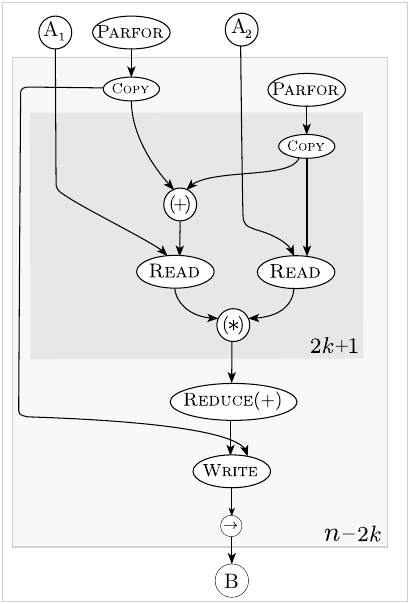}
\caption{1-D cross-correlation of a size $n$ array $A_1$ and a size $k$ array $A_2$. }
\label{fig:examples:conv}
\end{subfigure}
\caption{Example parametric graph templates for execution graphs. The edges are drawn in their input order left-to right.}
\label{fig:example:programs}
\end{figure}

\subsection{Syntax}
The vertices of the template graph are annotated with types corresponding to their function in the program. 
Each template graph contains the \emph{input memory vertices} $A_1, \dotsc,  A_k$, the \emph{output memory vertex} $B$, and the \emph{temporary memory vertices} $C_1, \dotsc, C_{k'}$. The memory vertices can have arbitrary in-degree and out-degree and belong to the root template.
Other vertices have out-degree $1$, except if stated otherwise. The outgoing edge is called the \emph{output edge}. To disambiguate the inputs to a vertex, the incoming edges are numbered consecutively. We refer to inputs in this \emph{input order}.
We consider the following control flow constructs. These are boundary vertices.

\miniheading{\textsc{Parfor}} A \textsc{Parfor} (parallel for loop) vertex has no input edge. Its output edge leads to a vertex in a child template.

\noindent
\textbf{\textsc{Reduce(\textsc{Op})}}, where \textsc{Op} is an associative and commutative operator. Has a single input edge from a vertex in a child template. The output edge leads to a non-memory vertex.

\miniheading{\textsc{Copy}} A \textsc{Copy} vertex $v$ has arbitrary outdegree. For each of its output edges $(v, u)$, the template $T(u)$ is not a parent of $T(v)$ and $u$ is not a memory vertex.

We consider the following memory constructs, which are boundary vertices.

\noindent
`$\mathbf{\rightarrow}$'. \ A (pass-through) $\rightarrow$ vertex has in-degree and out-degree $1$. At most one of the two neighbors can be a memory vertex.

\miniheading{\textsc{Read}} A Read vertex has a first input edge from a memory vertex or a $\rightarrow$ vertex and one or more other input edges from a non-memory vertex. Its output edge leads to a non-memory vertex.

\miniheading{\textsc{Write}} A Write vertex has two or more input edges from non-memory vertices. Its output edge leads to a memory or $\rightarrow$ vertex.

We consider the following types of operator vertices. They cannot be connected to memory vertices and are not boundary vertices.

\noindent
\textbf{$($\textsc{Op}$)$}, for \textsc{Op} $\in \{ +, -, *, \div \}$, which has in-degree $2$.

\noindent
{$\mathbf {[c]}$}, for any representable constant $c$.

\subsection{Semantics}


A \emph{well-formed} program has an acyclic template graph. For a well-formed program, a \emph{serial execution} is any topological order of the instantiation of the parametric graph template. Each $d$-dimensional input array $A_i$ initially contains some current value $A_i[j_1]....[j_d]$ at each position $(j_1, \dotsc, j_d)$, where the $d$-th dimension of $A_i$ has size $\text{size}_d(A_i)$. 
All other arrays contains $0$ at each of their positions. The arrays do not alias each other.

The semantics of a serial execution is given by applying the following rules to each vertex in the serial execution. Before evaluating the rules, contract all edges which have at least one $\rightarrow$ vertex neighbor (these exist to transfer values from inside the template hierarchy to the memory vertices in the root template).

\miniheading{\textsc{Parfor}} A \textsc{Parfor} vertex with $k$ children outputs one value of the permutation of $\{0, \dotsc, k-1\}$ to each child in an injective way.

\miniheading{\textsc{Copy}} Given input $x$, \textsc{Copy} outputs $x$ to all its children.

\noindent
\textbf{\textsc{Reduce(\textsc{Op})}}, for \textsc{Op} $\in \{ +, *\}$. Given the inputs $x_1, \dotsc, x_k$,  \textsc{Reduce(\textsc{Op})} outputs the result of applying \textsc{Op} repeatedly in an arbitrary order to the inputs.

\miniheading{\textsc{Read}} Given inputs $A_i, j_1, \dotsc, j_d$, if $A_i$ has dimension $d$ and for all $j_k\in \{j_1, \dotsc, j_d\}$ we have $0\leq j_k<\text{size}_k(A_i)$, a \textsc{Read} vertex outputs the current value of $A_i[j_1]....[j_d]$. Otherwise, the result of the serial execution is undefined.

\miniheading{\textsc{Write}} Given inputs $x, j_1, \dotsc, j_d$, a \textsc{Write} vertex outputs $x$. This has the side effect of updating the current value of the array into which the output edge leads: Say it leads to $A_i$. Then, if $A_i$ has dimension $d$ and for all $j_k\in \{j_1, \dotsc, j_d\}$ we have $0\leq j_k<\text{size}_k(A_i)$, the current value of $A_i[j_1]....[j_d]$ becomes $x$. Otherwise, the result of the serial execution is undefined.

\noindent
\textbf{$($\textsc{Op}$)$}, for \textsc{Op} $\in \{ +, -, *, \div \}$. Given $x,y$, outputs $x$ \textsc{Op} $y$.

\noindent
$\mathbf {[c]}$ outputs the constant $c$.

In a serial execution, we say that two reads or writes $u,v$ are \emph{totally ordered} if there is a path from $u$ to $v$ or from $v$ to $u$ in the instantiation. A \emph{data race} occurs if there are two writes $W_1, W_2$ with the same right input (i.e., index) that connect to the same memory vertex and $W_1$ and $W_2$ are not totally ordered. The output of a well-formed program is \emph{well-defined} if none of its serial executions has an undefined result or contains a data race. 

\subsection{Applications of Flows and Cuts}

To model dataflow in the parallel loop programs, we set the weight of the \textsc{Parfor} edges to $0$ and the weight of the other edges to $1$. Loop indices can be recomputed and thus do not cause data movement.
The parallel loop graph template encodes all the data movement in its edges. However, it cannot resolve the aliasing of array locations. Hence, the weight of the edges going across a partition of the vertices provides an \emph{upper}bound on the data movement:
\begin{observation}
Consider a partition $(V_0, , \dotsc, V_p)$ of the vertices in an instantiation of a parallel loop graph template. The value total weight of the edges with endpoints in different partitions is an upper bound on the data movement incurred when the partitions are allocated to distinct processors.
\end{observation}
Note that in our formulation of parallel loop graph templates, vertices corresponding to arrays are placed on a single processor. Thus, to model the distribution of an array across multiple processors, a vertex must be created for each processor that holds its subarray (this subarray can be discontiguous though).

Since a maximum $s$-$t$ flow equals the value of a minimum $s$-$t$ cut, the maxflow provides a partition of the loop program with small data movement:
\begin{observation}
If a parallel loop graph template has a maximum all-$s$-$t$ flow of value $x$, then there is a partition of the parallel loop program which incurs at most $x$ data movement and in which all instances of $s$ are executed on a different processor as all instances of $t$.
\end{observation}
We can get a similar statement for maximum single-$s$-$t$ flows.

%% file: sec-stflow.tex
\section{Template Maximum Flows} \label{sec:stflow}


Next, we turn to the first algorithmic question on parametric graph templates. Our goal here is to solve the maximum $s$-$t$ flows problem on a parametric graph template \emph{without explicitly instantiating it}. Instead, the goal is to get a runtime that is polynomial in the size of the graph template. Our algorithms use a series of observations on the structure of maximum flows in parametric graphs which allow us to produce transformed parametric graph templates, on which the answer can be efficiently computed.


We will approach the problem by considering the case where $s$ and $t$ are in the root template first. Then, we show how to reduce both the maximum all-$s$-$t$ flow and the maximum single-$s$-$t$ flow problem to an instance of this simpler problem. Throughout, we assume that all vertices are reachable from $s$ and can reach $t$, as otherwise they cannot carry flow. 

In the template-acyclic case, the maximum single-$s$-$t$ flow is trivially zero \emph{except} when $s$ and $t$ are in the same instance of the least common ancestor of $T(s)$ and $T(t)$ in the template tree. Therefore, in the acyclic case it makes sense to restrict our attention to this case where the flow is not trivially zero. In the case where there are template-cycles, it matters which instances of $s$ and $t$ are picked. These can be identified by numbering the instances they belong to.
 


\subsection{Edge-Reweighting}\label{sec:tech:edge-reweight}

\begin{figure}
\vspace{-2em}
\centering
\includegraphics[width=0.7\linewidth]{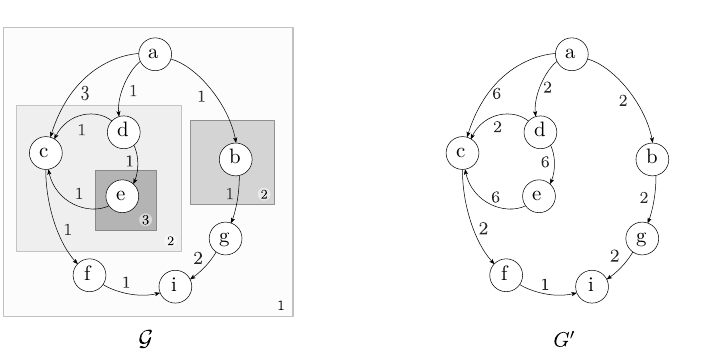}
\vspace{-1em}
\caption{Edge Reweighting turns a parametric graph template $\mathcal{G}$ into a graph $G'$ by scaling the weights of each edge in its template graph by all the parameters of the templates that contain an endpoint of the edge. }\label{fig:tech:reweight}
\end{figure}

An efficient way to solve a problem on parametric graph templates is to show how it relates to a problem on the template graph with \emph{scaled weights}. The idea is that an edge that intersects template $T_i$ can be used $P_i$ times and can therefore be used to carry $P_i$ times the amount of flow. We will see that this observation holds as long as $s$ and $t$ are in the root template or if we consider the maximum all-$s$-$t$ flow problem.
 We call this approach \emph{Edge-Reweighting}. 
See \Cref{fig:tech:reweight} for an example of Edge Reweighting. 

\miniheading{Algorithm: Edge-Reweighting}
Transform the parametric graph template $\mathcal{G} =(G,\mathcal{T}, \mathcal{P})$ with edge weights $w$  
into a graph $G'$ with edge weights $w'$. The \emph{reweighted graph} $G'$ has the same vertex and edge set as the template graph $G$, but the weights are scaled as follows:
Multiply the weight of an edge in the template graph by the product of the parameters of the templates that contain at least one endpoint of the edge. That is, let $I(e)$ be the index set of all templates that contain at least one of the endpoints of $e$. Then, the weight of $w'(e)$ is $w(e) \prod_{i\in I(e)} P_i$. To implement this in linear time $O(m)$, precompute in a pre-order traversal of the template tree for each template the product of all the ancestors' parameters.

\subsection{Source \& Sink Belong to the Root Template}\label{sec:maxflow:root}

Our goal is to show that when the source $s$ and sink $t$ belong to the root template, then a maximum $s$-$t$ flow in the reweighted graph equals the value of a maximum all-$s$-$t$ flow.
If $s$ and $t$ belong to the root template (which is instantiated once), then a maximum single-$s$-$t$ flow equals a maximum all-$s$-$t$ flow and we call it a maximum $s$-$t$ flow for short.

The linear programming \emph{dual} of a maximum $s$-$t$ flow is a \emph{minimum $s$-$t$ cut}~\cite{DBF55}. We will use strong duality~\cite{chvatal1983linear} in our proof, which means that it suffices to identify an $s$-$t$ flow and a minimum $s$-$t$ cut of equal value to prove that they are optimal. We argue that Edge 
	Reweighting preserves the value of the dual minimum all-$s$-$t$ cut. Hence, it also preserves  the maximum all-$s$-$t$ flow value. 

The following shows us how to construct an $s$-$t$ cut $C'$ in the transformed graph $G'$ from an $s$-$t$ cut $C$ in $\mathcal{G}$ of the same value. Together with the other (easier) direction of the proof, this shows that the transformed graph $G'$ has the same maximum $s$-$t$ flow.

\begin{lemma}\label{lem:stcut-structure}
In a parametric graph template $\mathcal{G}=(G,\mathcal{T}, \mathcal{P})$, if $s$ and $t$ are in the root template, there is a minimum $s$-$t$ cut of the instantiation of $\mathcal{G}$ where every instance of every vertex is on the same side of the cut.
\end{lemma}

\begin{proof}
The proof is by induction on the number of templates in the parametric graph template. If the parametric graph template has only a single template, then (since $s$ and $t$ must be in this template) the claim is trivial because the root is repeated only once, by assumption.

Otherwise, let $C=(V_s, V_t)$ be a minimum all-$s$-$t$ cut of the parametric graph template $\mathcal{G}$ (i.e., $V_s$ contains the vertices assigned to $s$ and $V_t$ those assigned to $t$).
Consider an arbitrary template $T_i$ that is a \emph{child of the root template} and its graph template $G_i$.
The sets $B_s$ and $B_t$ contain the boundary vertices of $T_i$ that are in $V_s$ and $V_t$, respectively. 

If either of the sets $B_s$ or $B_t$ is empty, then it follows immediately that all instances of the vertices that are in $T_i$ are in the same part of the cut $C$ (i.e., on the side of $s$ if the set $B_t$ is empty and vice versa). 

Otherwise, merging all vertices in $B_s$ into a vertex $s'$ and merging all vertices in $B_t$ into a vertex $t'$ does not change the value of the minimum $s$-$t$ cut in $\mathcal{G}$. Moreover, if the merged parametric graph template has a minimum $s$-$t$ cut that puts every instance of every vertex on the same side of the cut, then so does the original parametric graph template (because $B_s$ and $B_t$ contain only vertices that belong to the root template and we can ``undo'' the merging). We thus further assume w.l.o.g. that $B_s$ and $B_t$ contain a single vertex named $s'$ and $t'$, respectively. Note that since these vertices belong to the root template, the vertices $s'$ and $t'$ coincide with their only instances.

Every minimum $s$-$t$ cut must separate $s'$ from $t'$ in the subgraph $H$ given by the instances of $T_i$ and the vertices $s'$ and $t'$ (but without a potential edge from $s'$ and $t'$). We use our induction hypothesis to show there is a minimum $s'$-$t'$ cut in this subgraph that puts all instances of a vertex on the same side of the cut. 

We construct a parametric graph template $\mathcal{G}''$ such that a maximum $s'$-$t'$ flow in $\mathcal{G}''$ can be extended to a flow for the graph $H$. The parametric graph template $\mathcal{G}''$ has the following template graph: take the subgraph of $G$ induced by $T_i$ together with its boundary vertices, then delete any edges going between $s'$ and $t'$. The boundary vertices $s'$ and $t'$ and all vertices that belong to $T_i$ are put into the root template of $\mathcal{G}''$ (which has parameter $1$). Moreover, $\mathcal{G}''$ has the templates and parameters of the descendants of $T_i$ in $G$. The parametric graph template $\mathcal{G}''$ contains at least one template less than $G$. Hence, by induction, there is a minimum $s'$-$t'$ cut $C''$ of $\mathcal{G}''$ that puts all instances of the same vertex into the same partition. 

Let $f''$ be the dual maximum $s'$-$t'$ flow corresponding to $C''$ in $\mathcal{G}''$ of value $\mu$. Now, we construct a $s'$-$t'$ flow $f$ in $H$ and show it is maximum. Along each instance of each edge $e$ that intersects $T_i$ we send $f''(e)$ flow. The capacity constraint on the flow is trivially satisfied. The conservation constraint on the flow is satisfied because in the instantiated graph, the total flow going in and out of an instance of $v$ is the same as for vertex $v$ for $f''$ in $\mathcal{G}''$. The value of the flow $f$ is $P_i \cdot \mu$.
Now, consider the cut $C'$ where we put every instance of a vertex $v$ in $T_i$ on the same side as $v$ is in $C''$. The value of this cut is $P_i \cdot \mu$. By strong duality, this shows that $C'$ is a minimum $s'$-$t'$ cut in the graph $H$. By construction, this cut puts every instance of every vertex on the same side of the cut.
We conclude that all children of the root template can be cut such that every instance of the same vertex is in the same part of the cut. Because the root has a single instance, the statement follows for the root as well.
\end{proof}

\begin{lemma}\label{lem:stcut-reweighting}
	If a parametric graph template $\mathcal{G}$ has a minimum $s$-$t$ cut of value $\mu$ and $s$ and $t$ are in the root template, then the graph $G'$ constructed by edge reweighting has a minimum $s$-$t$ cut of value $\mu$.
\end{lemma}

\begin{proof}
Any cut in the reweighted graph $G'$ corresponds to a cut of the same value in the parametric graph template $\mathcal{G}=(G,\mathcal{T}, \mathcal{P})$: Put every instance of a vertex into the partition that it has in the cut of $G'$. Since every edge $e$ is cut exactly $\prod_{i\in I(e)} P_i$ times, this shows that the value of the minimum $s$-$t$ cut of the graph $G'$ is at least the value of the minimum all-$s$-$t$ cut of the parametric graph template $(G,\mathcal{T}, \mathcal{P})$.
It remains to show that the minimum $s$-$t$ cut of the reweighted graph $G'$ is at most the value of the minimum $s$-$t$ cut of the parametric graph template $\mathcal{G}$.
By \Cref{lem:stcut-structure}, there is a minimum $s$-$t$ cut $C$ of $\mathcal{G}$ that puts every instance of every vertex on the same side of the cut. Now, we construct a cut $C'$ of $G'$ from this cut $C$ by putting every vertex $v$ in $G'$ on the same side of the cut as all the instances of $v$ are in $C$. 
The cut $C'$ has the same value $\mu$ because every instance of an edge $e$ in $\mathcal{G}$ is crossing $\prod_{i\in I(e)} P_i$ times, which is the amount by which we scaled the weight of edge $e$ in $G'$. 
\end{proof}

%


\subsection{Instance Merging}\label{sec:tech:instance-merging}

We show how to merge all instances of a vertex $v$ in a parametric graph template by transforming it into parametric graph template of almost the same size (the overhead is an additive $O(nh)$). We will use this technique to reduce the general case for maximum all-$s$-$t$ flow to the case where $s$ and $t$ are in the root.

The idea is that merging all instances of a vertex $s$ is akin to moving the vertex from the template $T(s)$ it belongs to into the root template (so that it belongs to the root template). The \emph{no jumping rule} only allows edges to go from parent templates to children templates (or vice versa), we need to introduce \emph{dummy edges} and \emph{dummy vertices} along the way. The dummy edges have $\infty$ weight. An original edge $(u, s)$ will be transformed into a path $u, d_1, \dotsc, d_k, s$ for dummies $d_1, \dotsc, d_k$ (symmetrically for an edge $(s, u)$). See \Cref{fig:tech:im} for an example.

\begin{figure}
\begin{center}
	\includegraphics[width=0.5\linewidth]{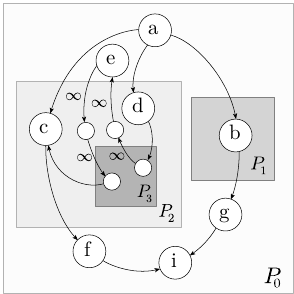}
\end{center}
\caption{Instance Merging on the graph $\mathcal{G}$ from \Cref{fig:intro-example:template} with vertex $e$ pushes the vertex $e$ is into the root template. This introduces dummy nodes (drawn without labels) and $\infty$-weight dummy edges. After contracting all dummy edges, the instantiation of the transformed parametric graph template is isomorphic to the graph we get by merging all instances of $e$ in the instantiation of $\mathcal{G}$.}\label{fig:tech:im}
\end{figure}

\miniheading{Algorithm: Instance-Merging} Given a parametric graph template $\mathcal{G}$ and a vertex $s$, repeat the following until $s$ is in the root template:
\begin{enumerate}
	\item For any cross-template edge $(u, s)$, introduce a dummy vertex $d$ in the template $T(s)$ that $s$ belongs to. Replace the edge $(u, s)$ by two edges $e_1=(u, d)$ and $e_2=(d, s)$. The weight of the edge $e_1$ is the same as the weight of the edge $e$, but the weight of the edge $e_2$ is set to $\infty$. Proceed symmetrically for any cross-template edge $(s, u)$.
	\item Move the vertex $s$ from the template $T(s)$ to the parent of the template $T(s)$ (i.e., remove $s$ from the set $T(s)$). 
\end{enumerate}

\begin{observation}\label{lem:tech:instance-merge}
	Instance Merging($\mathcal{G}$, $s$) produces a parametric graph template $\mathcal{G'}$ whose instantiation is, after merging all instances of dummy edges of weight $\infty$, isomorphic to the graph that we get by instantiating the original parametric graph template $\mathcal{G}$ and merging all instances of $s$. Instance Merging($\mathcal{G}$, $s$) adds at most $d(s) \cdot h$ vertices and edges, where $d(s)$ is the degree of the vertex $s$ in the template graph. 
\end{observation}
\begin{proof}
	In the template graph of $\mathcal{G'}$, there is a path consisting of $\infty$-weight edges from $s$ to every neighbor of $s$ in the template graph $G$ of $\mathcal{G}$. There are no other $\infty$ weight edges. Hence, there also is such an $\infty$-weight path in the instantiation of $\mathcal{G'}$ to every instance of every vertex that is a neighbor of $s$ in $G$. Contracting these paths gives a graph where the vertex $s$ has an edge to all instances of neighbors of $s$ in $G$, which is the same graph that we get by instantiating the original parametric graph template $\mathcal{G}$ and merging all instances of $s$. 
\end{proof}

\subsection{Maximum All-$s$-$t$ Flow}\label{sec:allstflow}

To solve maximum all $s$-$t$ Flow, all we would need to do is use Instance Merging on $s$ and then on $t$ to ensure that they are both in the root template. Then, we could use the edge reweighing \Cref{lem:stcut-reweighting}. This approach would cost  $O( n m + n^2 h)$ time. We can avoid this overhead by observing that edge reweighting works directly for maximum all-$s$-$t$ Flow (even when $s$ and $t$ are not in the root template).

\begin{lemma}\label{thm:edge-reweighting}
Edge reweighting of a parametric graph template $\mathcal{G}$ produces a re-weighted graph $G'$ where the value of the maximum $s$-$t$ flow of $G'$ equals the value of the maximum all-$s$-$t$ flow of $\mathcal{G}$.
\end{lemma}
\begin{proof}
Instance Merge $s$ and then $t$ in $\mathcal{G}$ to produce a parametric graph template $\mathcal{G'}$. By definition, all instances of $s$ (and $t$ respectively) must be on the same side of a minimum all-$s$-$t$ cut, this parametric graph template $\mathcal{G}'$ has the same minimum all-$s$-$t$ cut value as the original parametric graph template $\mathcal{G}$. Edge reweighting $\mathcal{G'}$ gives us a graph $\hat G$. From \Cref{lem:stcut-reweighting} we know that a minimum $s$-$t$ cut of $\hat G$ corresponds to the minimum all-$s$-$t$ cut of $\mathcal{G'}$ (which puts all instances of the same vertex on the same side of the cut).

An $\infty$-weight edge never crosses a minimum $s$-$t$ cut and therefore such dummy edges (introduced by the instance merging) from $\hat G$ can be contracted, yielding a graph $G'$. This graph $G'$ is the same graph that we get from edge reweighting the original parametric graph template $\mathcal{G}$.
\end{proof}
Now, the results follows:
\begin{theorem}\label{thm:allstflow}
	Computing a maximum all-$s$-$t$ flow in parametric graph template takes $O(mn)$ time.
\end{theorem}

\subsection{Partial Instantiation}\label{sec:tech:partial-instantiation}

\begin{figure}[t!]
\includegraphics[width=0.45\linewidth]{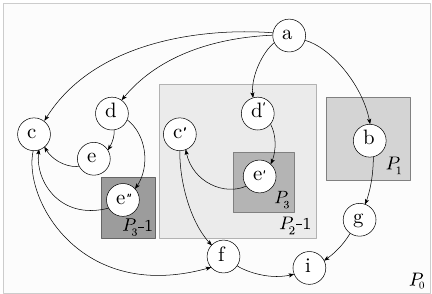}
\centering
\caption{After running Upwards Partial Instantiation from $e$ on the parametric graph template $\mathcal{G}$ from \Cref{fig:intro-example:template}, the vertex $e$ is in the root template. The transformed parametric graph template instantiates the same graph.}
\label{fig:tech:part}
\end{figure}

The technique of \emph{partial instantiation} revolves around instantiating only part of the parametric graph template, depending on the problem at hand. The goal is to choose the partial instantiation such that the remaining problem is solvable by using the symmetry of the problem (e.g., using edge-reweighting). 
Partial instantiation can be seen as an example of the more general technique of \emph{retemplating}. The intuition of retemplating is that in certain cases, it suffices to change the representation of the parametric graph template into another isomorphic parametric graph template to significantly simplify the problem at hand.

Next, show how to move a single vertex $s$ from deep in the template tree to the root, without changing the instantiated graph. This solves the maximum $s$-$t$ problems when $s$ (or $t$) belongs to a template that is deep in the template tree (See \Cref{sec:maxflow-acyc}). We call this technique \emph{Upwards Partial Instantiation} from $s$. For simplicity, let us start with the special case of template-acyclic graphs.

In a template-acyclic parametric graph template, once a path goes from an instance of a template $T_i$ to its parent, it never enters another instance of $T_i$ again. This property implies that, when considering the reachable subgraph from a vertex that is an instance of $s$, we can simply ``merge'' $T(s)$ and all the templates that are ancestors of the template $T(s)$ in the template tree. Formally, this corresponds to deleting $T(s)$ and all the templates that are ancestors of $T(s)$ (except the root) from the parametric graph templates' list of templates.

If the parametric graph template has template-cycles, our goal remains to transform the parametric graph template into an equivalent graph where a particular instance of a vertex $s$ is in the root template. See \Cref{fig:tech:part} for an illustration of Upwards Partial Instantiation.

\miniheading{Algorithm: Partial Instantiation}
Repeat the following until all templates from $T(s)$ to the root have parameter $1$:
\begin{enumerate}
	\item Consider the topmost template $T$ that contains $s$ and has parameter greater than $1$. Let $P_s$ be the number of instances of the template $T$.  
	\item Instantiate the template $T$ twice. Create a new parametric graph template that has the two instances as templates, where the first template has parameter $1$ and the second template has parameter $P_s-1$. The vertices in the second template are relabeled ($s$ is in the one with parameter $1$).
\end{enumerate}
Now, merge $T(s)$ and all the templates that are ancestors of $T(s)$, leaving $s$ in the root template. 

Because this process performs the same rewriting of the parametric graph template as instantiation, just in a different order and stopping early, this process creates an isomorphic parametric graph template. Every iteration adds at most $n$ vertices and $m$ edges and there are at most $h$ iterations. We conclude that:
\begin{observation}
	Upwards Partial Instantiation from $s$ produces an isomorphic parametric graph template with at most $h$ additional templates and $O(nh)$ vertices and $O(mh)$ edges in the template graph.
\end{observation}

\subsection{Maximum Single-$s$-$t$ Flow}\label{sec:maxflow-acyc}

We give a partial instantiation and edge reweighting approach to maximum single-$s$-$t$ flow. For there to be a flow through some instance, it must lie along an $s$-$t$ path. Hence, we can use Upwards Partial Instantiation twice to ensure that $s$ and $t$ lie in the root template. Then, we use Edge Reweighting.

\miniheading{Algorithm: Single-s-t Flow} We solve maximum single-$s$-$t$ flow as follows:
\begin{enumerate}
	\item Perform upwards partial instantiation from $s$. 
	\item Perform upwards partial instantiation from $t$.
	\item Construct an edge-reweighted graph $G'$.
	\item Run a maximum $s$-$t$ flow algorithm on the partially instantiated and reweighted graph $G'$.
\end{enumerate} 

\begin{theorem}\label{thm:singlstflow}
	Computing a maximum single-$s$-$t$ flow in parametric graph template takes $O(mn h)$ time.
\end{theorem}


%% file: sec-sibling.tex
\section{Allowing Edges Between Sibling Templates}

So far, we have disallowed any edges between instances of the same template. This limits the types of graphs which have a small template graph. For example, a path of length $n$  requires a template graph with $n$ nodes.
We can extend the model by allowing an instance to have edges \emph{to another instance of the same template}. These edges can, for example, more efficiently model sequential chains (paths), convolutional networks, and grids. We call these edges \emph{sibling edges} (because they connect siblings).
A \emph{sibling edge} (u,v) of template $T_i$ connects (in the template graph of $T_i$) a vertex $u$ that belongs to template $T_i$ to a vertex $v$ that \emph{also} belongs to $T_i$. Every sibling edge  $e=(u,v)$ of template $T_i$ is associated with a \emph{bijective} (and computable) \emph{sibling function} $f_e : \{0, \dotsc, P_i-1\} \rightarrow \{0, \dotsc, P_i-1\}$ which tells us that if the head of edge $e$ is in instance $j$ of the template $T_i$, then the tail of edge $e$ is in instance $f(j)$ of the template $T_i$.
%

Note that in the model with sibling edges, a path of length $n$ can be represented with two nodes instead of $n$ nodes and a $1$-dimensional cross-correlation of two $n$-dimensional signal can be represented with $2$ nodes.

\begin{figure}[t]
\centering
	\includegraphics[width=0.5\linewidth]{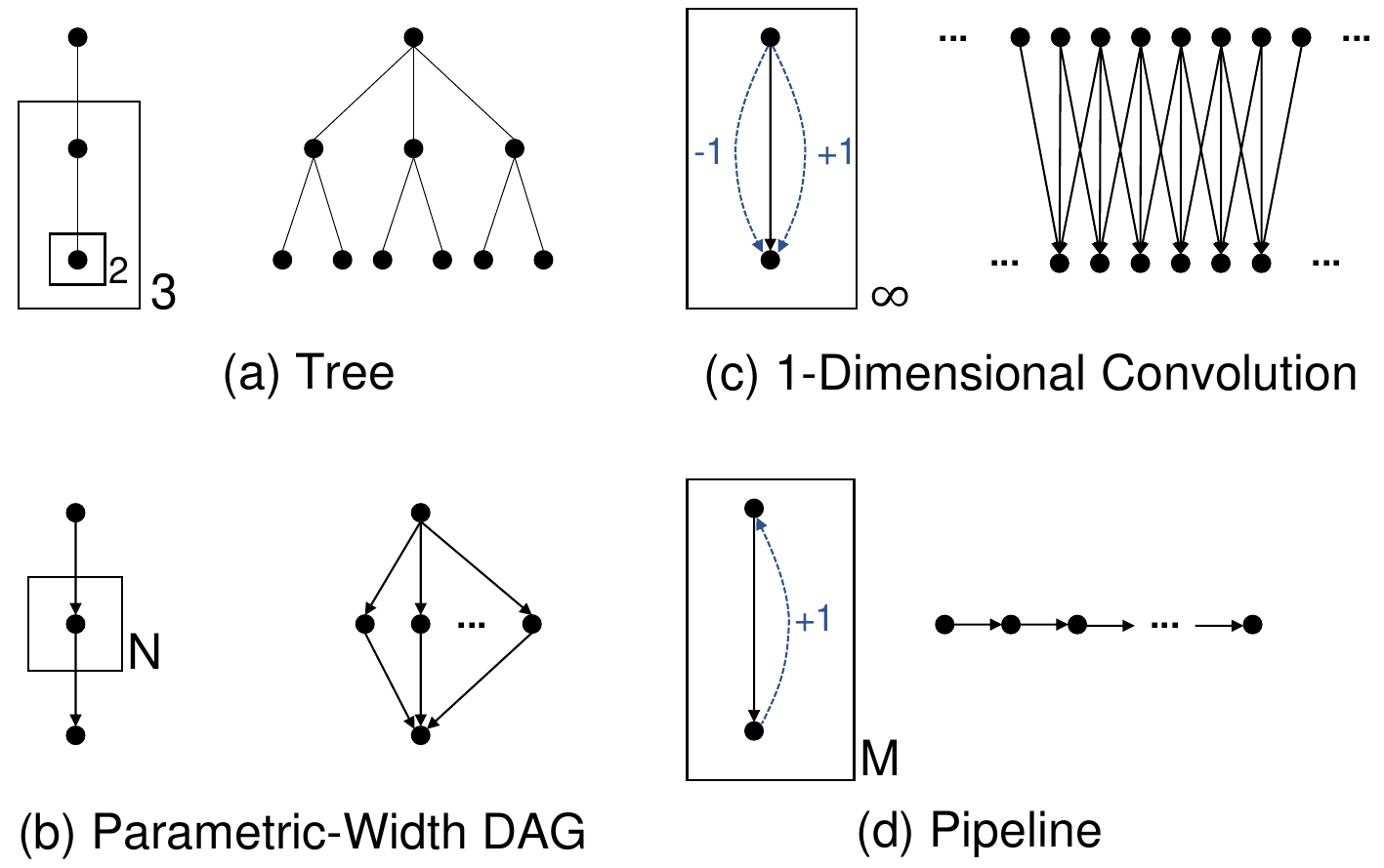}
	\vspace{-0.5em}
	\caption{Graph Templates and their expanded/instiantiated counterparts. An edge $e$ labelled with $\pm \Delta$ indicates a sibling edge $e$ with sibling function $f_e(x)=x+\Delta$}
	\label{fig:examples}
\end{figure}

The structural \Cref{lem:stcut-structure} for edge reweighting still holds with sibling edges.

\begin{lemma}\label{lem:stcut-structure-sibling}
In a parametric graph template $\mathcal{G}=(G,\mathcal{T}, \mathcal{P})$ \emph{with sibling edges}, if $s$ and $t$ are in the root template, there is a minimum $s$-$t$ cut of the instantiation of $\mathcal{G}$ where every instance of every vertex is on the same side of the cut.
\end{lemma}

\begin{proof}[sketch of Lemma 1 with sibling edges] Extend the proof of \Cref{lem:stcut-structure} as follows: Proceed to construct the flow $f$ as usual. Then, observe that the conservation constraint on the flow $f$ is satisfied because the sibling functions are bijective: In the instantiated graph, the total flow going in and out of an instance of $v$ is the same as for vertex $v$ for $f''$ in $G''$. The value of the flow $f$ is $P_i \cdot v(f)$.
\end{proof}

Hence, the results on maximum all-$s$-$t$ flow and maximum single-$s$-$t$ flow hold analogously in the presence of sibling edges within the same bounds.


%% file: sec-conclusion.tex
\section{Conclusion}
In this work, we explored the notion of structural parameterization in graphs. We show how graph templates correspond to the computation graphs of parallel programs. Our model leads to a $O(mn)$ time algorithm for a template version of maximum $s$-$t$ flow (and hence minimum $s$-$t$ cuts). These flows provide upper bounds on the data movement of partitions of certain parallel loop programs.

Other interesting problems would include partitions into multiple parts and subgraph isomorphism. Moreover, future work could explore lower bounds for parametric graph template algorithms.

